\documentclass{llncs}

\makeatletter
\let\proof\@undefined
\let\endproof\@undefined
\makeatother

\usepackage{enumerate}
\usepackage{url}
\usepackage{ulem}
\usepackage{a4wide}
\usepackage{color}
\usepackage{amsfonts}
\usepackage{amsthm}
\usepackage{amssymb}
\usepackage{amsmath}
\usepackage[pdftex]{graphicx}
\usepackage{epstopdf}
\usepackage{color}

\newcommand{\cT}{{\mathcal T}}

\newtheorem{observation}{Observation}

\author{Alexander Grigoriev \inst{1}, Steven Kelk\inst{2}, Nela Leki{\'c}\inst{2}}

\institute{
Department of Quantitative Economics, Maastricht University, P.O. Box 616, 6200 MD Maastricht, The Netherlands, \email{a.grigoriev@maastrichtuniversity.nl}
\and
Department of Knowledge Engineering (DKE), Maastricht University, P.O. Box 616, 6200 MD Maastricht, The Netherlands, \email{steven.kelk@maastrichtuniversity.nl, nela.lekic@maastrichtuniversity.nl}
}

\begin{document}
\title{On low treewidth graphs and supertrees}
\maketitle

\begin{abstract}
\noindent
 Compatibility of unrooted phylogenetic trees is a well studied problem in phylogenetics. It asks to determine whether for a set of $k$ input trees $T_1,...,T_k$ there exists a larger tree (called a supertree) that contains the topologies of all $k$ input trees. When any such supertree exists we call the instance compatible and otherwise incompatible. It is known that the problem is NP-hard and FPT, although a constructive FPT algorithm is not known. It has been shown that whenever the treewidth of an auxiliary structure known as the display graph is strictly larger than the number of input trees, the instance is incompatible. Here we show that whenever the treewidth of the display graph is at most 2, the instance is compatible. Furthermore, we give a polynomial-time algorithm to construct a supertree in this case. Finally, we demonstrate both compatible and incompatible instances that have display graphs with treewidth 3, highlighting that the treewidth of the display graph is (on its own) not sufficient to determine compatibility.
\end{abstract}

\section{Introduction}
One of the central challenges within computational evolutionary biology is to infer the evolutionary history of a set of contemporary species (or more generally, \emph{taxa})  $X$ using only the genotype of the contemporary species. This evolutionary history is usually modeled as a \emph{phylogenetic tree}, essentially a tree in which the leaves are bijectively labeled by the elements of $X$ and the internal nodes of the tree represent (hypothetical) ancestors \cite{SemSte03}.

There is already an extensive literature available on the extent to which different optimization criteria on the space of phylogenetic trees (e.g. likelihood, parsimony) are able to identify the ``true'' evolutionary history. In any case it is well-known that most of these optimization criteria are NP-hard, and this intractability is a serious problem when constructing phylogenetic trees for large numbers of taxa. This has been one of the motivations behind \emph{supertree} methods \cite{BinindaEmonds2004}. Here the goal is to first construct phylogenetic trees for small (overlapping) subsets of $X$ and then to puzzle the partial trees together into a single tree on $X$ that contains all the topologies of the partial trees, in which case we say the partial trees are \emph{compatible}, or to conclude that no such tree exists.

The computational complexity landscape of the compatibility problem is uneven. In the case that all the partial trees are rooted (i.e. in which the flow of evolution is assumed to be away from a designated root, towards the taxa) the problem is polynomial-time solvable, using the algorithm of Aho \cite{Aho81}. However, in the case of unrooted trees the problem is NP-hard, even when all the partial trees have at most 4 taxa \cite{Steel92}. Nevertheless, due to the fact that many tree-building algorithms actually construct \emph{un}rooted trees, and because of the risk of distorting the underlying phylogenetic signal through a poor choice of root location, it remains attractive to try and solve this NP-hard variant of the problem directly.

In this article we approach the unrooted compatibility problem from a graph-theoretical angle. There is a recent trend in this direction, which to a large extent can be traced back to a seminal paper of Bryant and Lagergren \cite{BryLag06}. They observed that there is a relationship between the compatibility question and the \emph{treewidth} of an auxiliary graph known as the \emph{display graph}. The display graph is obtained by identifying the taxa of the input trees, and treewidth is an intensely well-studied parameter in the algorithmic graph theory literature (see e.g. \cite{BodlaenderK10}). Low (or bounded) treewidth often facilitates algorithmic tractability, and given that it is a measure of ``distance from being a tree'', it is tempting to try and exploit this tractability in questions pertaining to phylogenetic compatiblity and incongruence. Bryant and Largergren observed that for $k$ unrooted trees to be compatible, it is necessary (but not sufficient) that the display graph has treewidth at most $k$. The upper bound on the treewidth that this condition generates, subsequently makes it possible to formulate and answer the compatibility question in a computationally efficient way. However, this efficiency is purely theoretical in nature, obtained via the indirect route of monadic second order logic \cite{courcelle1990monadic}, and it remains a challenge to succinctly characterize phylogenetic compatibility. Since Bryant and Largergren various other authors have picked up this thread (e.g. \cite{GruHum08}), with particular attention for triangulation-based approaches (see e.g. \cite{vakati2011graph,gysel2012reducing,VakBac13}) although the question remains: what \emph{exactly} is the role of treewidth in compatibility?

Here we take a step forward in understanding the link between treewidth and compatibility. We prove that if the display graph of a set of unrooted binary trees has treewidth at most 2, then the input trees are compatible, and this holds for any number of input trees. In other words, it is not necessary to look deeper into the structure of the display graph, compatibility is immediately guaranteed. The proof of this, based on graph separators and graph minors, is surprisingly involved. Moreover, we describe a simple polynomial-time algorithm to construct a supertree for the input trees, when this condition holds. We also show that in some sense this result is ``best possible'': we show how to construct both compatible and incompatible instances that have display graphs of treewidth 3, for any number of trees. This confirms that the treewidth of the display graph cannot, on its own, fully capture phylogenetic compatibility, and that auxiliary information is indeed necessary if we are to obtain a complete characterization. 

Clearly, the significance of this result does not lie in its immediate practical relevance.
Rather, the main contribution of this article is that it opens the door to the possibility that existing ``descriptive'' characterisations of compatibility (e.g. legal triangulations \cite{vakati2011graph}) can be specialized into simple and efficient combinatorial algorithms when the display graph has sufficiently low treewidth.

%



\section{Preliminaries}
Let $X$ be a finite set. An \emph{unrooted phylogenetic $X$-tree} is a tree whose leaves are bijectively labeled by the elements of set $X$. It is called \emph{binary} when all its inner nodes (nonleaf nodes) are of degree 3. An unrooted binary phylogenetic tree on four leaves is called a \emph{quartet}. In the remainder of the article we focus almost exclusively on
unrooted binary trees, often writing simply trees or $X$-trees for short.

We call elements of $X$ \emph{taxa} or \emph{leaves}. For some $X$-tree $T$ and some subset $X' \subseteq X$ we denote by $T(X')$ the subtree of $T$ induced by $X'$ and by $T|X'$ the tree obtained from $T(X')$ by suppressing vertices of degree 2. Furthermore, we say a tree $S$ \emph{displays} a tree $T$ if $T$ can be obtained from a subgraph of $S$
by suppressing vertices of degree two.

Given a set $X$
a \emph{split} is defined as a bipartition of $X$. If we label the components of the partition by $A$ and $B$, then we can denote the split by $A|B$. Note that each edge of an $X$-tree naturally induces
a split.
If $A|B$ is a split induced by an edge of a tree $T$, then we say that $T$ \emph{contains} split $A|B$.
We use $ab|cd$ to denote the quartet in which taxa $a$ and $b$ are on one side of the internal edge and $c$ and $d$ are on the other. We write $ab|cd \in T$ if $T$ displays $ab|cd$.

Given a set $\cT$ of $k$ trees $T_1,...,T_k$ we wish to know if there exists a single tree $S$ that displays $T_i$ for all $i \in\{1,...k\}$.
A tree that displays all the input trees, if such a tree exists, is called a \emph{supertree}. When a supertree does exist we call the instance \emph{compatible}, otherwise \emph{incompatible}. A supertree is not necessarily unique. To see when such a tree is unique and many more details on this topic we refer the reader to \cite{DreHub12} or \cite{SemSte03}.

The \emph{display graph} $D(\cT)$ of a set of trees $\cT$ is the graph obtained from the disjoint union of trees in $\cT$ by identifying vertices with the same taxon labels. Note that $D(\cT)$ can be disconnected if and only if the trees in $\cT$ can be bipartitioned into two sets $\cT_1, \cT_2$ such that $X(\cT_1) \cap X(\cT_2) = \emptyset$, where $X(T)$ refers
to the set of taxa of $T$. In such a case $\cT$ permits a supertree if and only if both $\cT_1$ and $\cT_2$ do. Hence for the remainder of the article we focus on the case when $D(\cT)$ is connected.

Before we can start discussing our result we need a few graph theoretic definitions.  Let $G=(V, E)$ be an undirected graph. For any two subsets of vertices $A,B \subseteq V$ and any $Z \subseteq V $
we say $Z$ \emph{separates} sets $A$ and $B$ in $G$ if every path in $G$ that starts at some vertex $u \in A$ and ends at some vertex $v \in B$ contains a vertex
from $Z$. Such a set $Z$ is called an $(A,B)$-\emph{separator}, or simply  a \emph{separator}. A graph $M$ is a \emph{minor} of a graph $G$ if $M$ can be obtained from a subgraph of $G$ by contracting edges.
The treewidth of a graph $G$, denoted $tw(G)$, has a somewhat technical definition and we refer to e.g. \cite{BodlaenderK10} for details. For the main result it is sufficient to note that trees have treewidth 1, and
that graphs with treewidth at most 2 are exactly those graphs that do not have a $K_4$-minor (where $K_4$ is the
complete graph on 4 vertices). We will also use
the well-known fact that if $M$ is a minor of $G$, $tw(M) \leq tw(G)$.
 For remaining graph theory terminology we refer to standard texts such as \cite{diestel2000graph}.



\section{Main results}

We begin with some simple lemmas.

\begin{lemma}
\label{conflict_quar}
\textbf{\emph{\cite[Corollary 1]{Gan02}}}. Let $T_1$ and $T_2$ be two unrooted  phylogenetic trees on the same set of taxa $X$. Then $T_1$ and $T_2$ are compatible if and only if there do not exist four taxa ${a,b,c,d} \subseteq X$ such that $ab|cd \in T_1$ and $ac|bd \in T_2$. 
\end{lemma}

\begin{lemma} 
\label{display_minor}
Let $D$ be the display graph of the two quartets $ab|cd$ and $ac|bd$. Then $D$ has $K_4$ as a minor. Hence, $tw(D) \geq 3$.
\end{lemma}
\begin{proof}
 Let $Q_1=ab|cd$ and $Q_2=ac|bd$. Both $Q_1$ and $Q_2$ have exactly two inner nodes, denote them $u, v$ and $w,z$ respectively. Then it is immediate to see that vertices $u,v,w,z$ of $D$ form a $K_4$ minor (obtained by suppressing leaves $a,b,c,d$ which all have degree 2 in $D$).
\end{proof}

\begin{theorem}
 \label{thm:2trees}
 Let $T_1$ and $T_2$ be two unrooted phylogenetic trees. Let $D$ be the display graph of $T_1$ and $T_2$. Then $T_1$ and $T_2$ are compatible if and only if $tw(D)\leq 2$.
\end{theorem}

\begin{proof}
 Let $T_1$ and $T_2$ be two trees on taxa sets $X$ and $X'$ respectively. Let $X^* = X \cap X'$. Then $T_1$ and $T_2$ are compatible if and only if $T_1|X^*$ and $T_2|X^*$ are compatible \cite{SemSte03}. Thus we only have to consider two  trees $T_1$ and $T_2$ on the same set of taxa $X$. Let $D(T_1,T_2)$ be their display graph. Suppose for the sake of contradiction that $tw(D(T_1,T_2)) \leq 2$ while $T_1$ and $T_2$ are incompatible. From Lemma \ref{conflict_quar}, $T_1$ and $T_2$ contain incompatible quartets $Q_1$ and $Q_2$ (w.l.o.g. let $T_i$ display $Q_i$) and since $Q_i$ is displayed in $T_i$, it is also displayed in $D(T_1,T_2)$, so $D(Q_1,Q_2)$ is a minor of $D(T_1,T_2)$. 
Since $D(Q_1,Q_2)$ is a minor of $G$, and using Lemma \ref{display_minor}, $tw(G)\geq tw(D(Q_1,Q_2))\geq 3$,
contradicting the fact that $tw(D(T_1,T_2)) \leq 2$. This completes our proof in one direction; for the other see \cite{BryLag06}. 
\end{proof}

In the following main theorem we emphasize that the trees in $\cT$ do not need to be on the same set of taxa, but that for this proof the input trees do need to be binary.
\newpage

\begin{theorem} \label{ktrees}
 Let $\cT$ be a set of $k$ binary unrooted phylogenetic trees $T_1, ..., T_k$ and let $D$ be their display graph. If
$tw(D) \leq 2$, then $T_1, ..., T_k$ are compatible, in which case a supertree can be constructed in polynomial time.
\end{theorem}

\begin{proof}

We give a constructive proof in which we will build a supertree $S$ for $\cT$. The idea is to find an appropriate separator of $D$ and to reduce the problem into smaller instances of the same problem i.e. an induction proof. The
induction will be on the cardinality of $X = \cup_{T_i \in \mathcal{T}} X(T_i)$. For the base case observe that
an instance with $|X| \leq 3$ is trivially compatible. 


Before we start the construction we apply a number of operations on $D$ that are safe to do, in sense that they preserve (in)compatibility of the instance and do not cause the treewidth of $D$ to rise. We remove any taxon that has degree 1 in $D$ and contract any inner vertex that has degree 2 in $D$. This clearly affects neither the compatibility nor the treewidth. Furthermore, whenever we encounter a tree with strictly less than 4 taxa, we can remove it from $D$. Such a tree carries no topological information and thus does not change the compatibility, while removing something from a graph cannot increase its treewidth. The \emph{cleaning up} procedure means that we apply all these operations on $D$ repeatedly until we cannot apply any anymore. In other words, we can assume $D$ to have treewidth exactly 2, that all inner vertices of $D$ have degree 3, that all taxa have degree at least 2 and that no tree has fewer than 4 taxa.

Consider a planar embedding of the display graph $D(\cT)$. This exists and can be found in polynomial time because $D(\cT)$ is planar. The \emph{boundary} of a face $F$ of $D(\cT)$, denoted $B(F)$, is the set of edges and vertices that are incident to the interior of the face. We say that two distinct faces $F_1, F_2$ are \emph{minimally adjacent} if the following three conditions hold: (1) $F_1$ and $F_2$ are adjacent; (2) $B(F_1) \cap B(F_2)$ is isomorphic to a path containing at least one edge; (3) the internal vertices of the path $B(F_1) \cap B(F_2)$ all have degree 2 in $D(\cT)$, and the two endpoints of the path each have degree 3 or higher in $D(\cT)$. See the appendix for a proof that if the treewidth of $D$ is 2 we can always find two such faces, neither
equal to the outer face, in polynomial time.

Let $F_1$ and $F_2$ be two minimally adjacent faces of $D$, neither equal to the outer face. Denote by $p(u,v)$ the path $B(F_1) \cap B(F_2)$ they share. (After locating $F_1$ and $F_2$ this path can easily be found in polynomial time). By definition $u$ and $v$ must have degree at least 3 in $D$. Also, by minimal adjacency of $F_1$ and $F_2$ and due to cleaning up, none of the interior nodes of $p(u,v)$ can be internal tree nodes. Moreover, since we removed all trees on fewer than four taxa, at most one leaf can appear as an interior node of the path. Such a leaf can only exist if both $u$ and $v$ are inner nodes of some trees. Now, $u$ and $v$ can either be both leaves, both inner nodes or one of them a leaf another an inner node. These are the three cases we have to consider.\\
\\
\textbf{Case(i)} is when both $u$ and $v$ are leaves. We claim this cannot happen. In this case, path $p(u,v)$ must be an edge. But if it is an edge it is connecting two leaves and will have already been removed during cleaning up.\\
\\
\textbf{Case(ii)} is when $u$ is a leaf and $v$ is an inner node. Again we have that path $p(u,v)$ must be an edge $(u,v)$ which both faces share. Let $x$, respectively $y$, be any vertex other than $u$ or $v$ on the boundary of $F_1$, respectively $F_2$. See Figure \ref{fig:K4separator}(a). We claim that any path between $x$ and $y$ must contain either $u$ or $v$. In particular, suppose there exists a path $p(x,y)$ such that $u,v \notin p(x,y)$. Let $x'$ and $y'$ be vertices on $p(x,y)$ such that the subpath $p(x',y')$ is the shortest subpath of $p(x,y)$ with the property that both of its endpoints are on the boundaries of $F_1$ and $F_2$, respectively. See Figure \ref{fig:K4separator}(a). Then $D$ contains a $K_4$ minor formed by vertices $u, v, x'$ and $y'$. 
This is a contradiction on $D$ having treewidth 2. So we have that any path between $x$ and $y$ passes through either $u$ or $v$. Thus $\{u,v\}$ is a separator of $D$. 

\begin{figure}[h]
 \begin{center}
  \includegraphics[width=10cm]{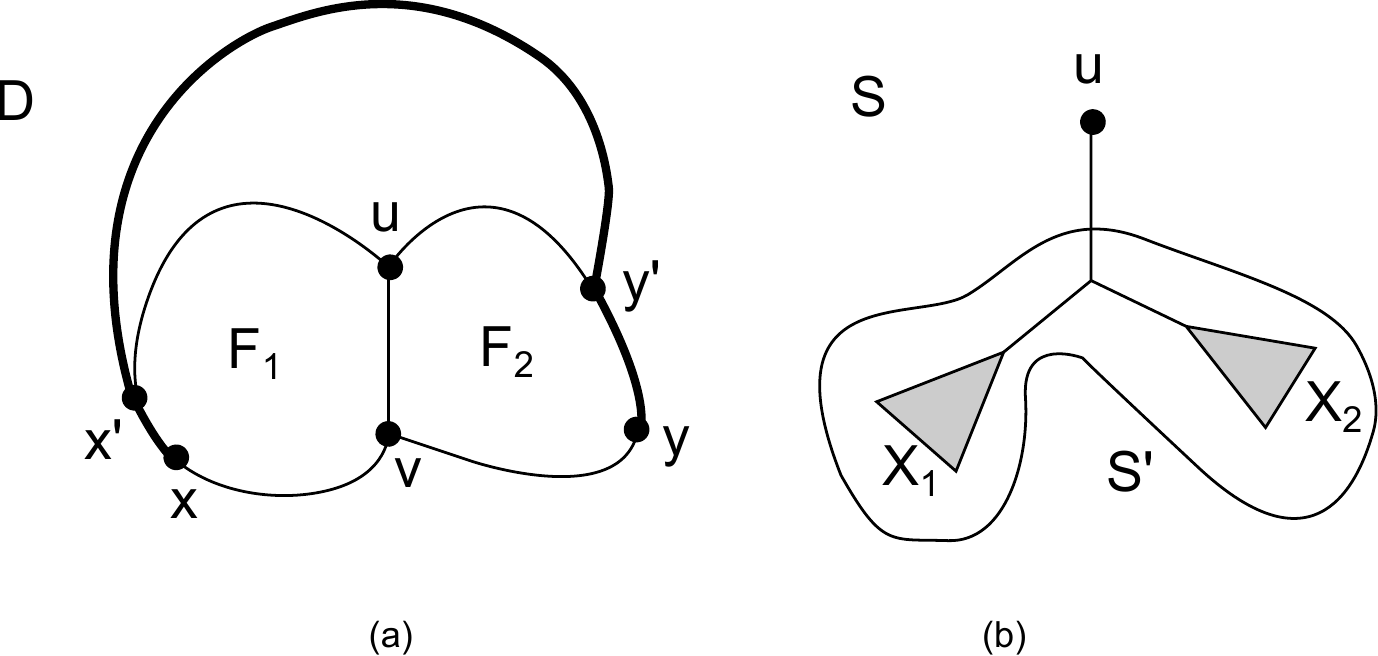}
 \end{center} 
 \caption{(a) Two minimally adjacent faces $F_1$ and $F_2$ in $D$. The vertices $u,v,x',y'$ induce a $K_4$ minor. (b) A supertree as constructed in case (ii).}
 \label{fig:K4separator}
\end{figure}

Removing $u$ and $v$ from the vertex set of $D$ disconnects it and divides the set of taxa into two sets $X_1$ and $X_2$, such that $X =  X_1 \cup X_2 \cup \{u\}$. 
We claim that supertree $S$ as shown in Figure \ref{fig:K4separator}(b), where $S'$ is a supertree of $T_1,...,T_k$ restricted to taxa set $X\setminus \{u\}$, displays all $k$ input trees $T_1,...,T_k$. To prove this we have to show two things. One, that the supertree $S'$ exists (and that it has an edge corresponding to split $X_1|X_2$) and two, that all quartets in $T_1,...,T_k$ are also in $S$. (The latter is sufficient because a set of unrooted trees is compatible if and only if the set of quartets displayed by the
trees is compatible).

To prove the first claim let $X':=X\setminus \{u\}$ and notice that by induction the instance $T_1,...,T_k|X'$ is compatible and thus has a supertree. We now claim that there exists some supertree of $T_1,...,T_k|X'$, call it $S'$, which contains split $X_1|X_2$. First of all notice that (a restriction of) $X_1|X_2$ must  be a split in every input tree restricted to $X'$. To see this we show that there does not exist a quartet $ab|cd$ with $a,c \in X_1$ and $b,d \in X_2$ in any of the input trees (prior to removal of $u$ and $v$). Suppose such a quartet did exist in some tree.
 Then there would exist edge-disjoint paths $p(a,b)$ and $p(c,d)$ in $D$, where the interior nodes of these paths are internal tree nodes. Since removing $u$ and $v$ from $D$ disconnects it (such that $X_1$ and $X_2$ are subsequently in separate components), it must be that those paths had to use either $u$ or $v$. Since $u$ is a taxon it cannot be used for this purpose. So both paths had to use inner vertex $v$. However, this contradicts the edge-disjointness of the two paths. Hence quartet $ab|cd$ cannot be displayed by any tree.

We conclude from this that in each $T_i|X'$ there exists an edge $e$ that induces a split $A|B$, such that $A \subseteq X_1$ and $B \subseteq X_2$. 
Furthermore both $X_1$ and $X_2$ must contain at least one taxon each. (This follows because edge $(u,v)$ belongs to some input tree $T$, and walking from $u$ to $v$ along the boundary of $F_1$ whilst avoiding edge $(u,v)$ necessitates
entering and leaving $T$ via its taxa, which in turn means that some taxon not equal to $u$ must exist on the part of the boundary of $F_1$ not shared by $F_2$. The same argument holds for $F_2$.) 
As such, in each $T_i|X'$ it is possible to contract (the subtree induced by) $X_1$ and/or $X_2$ into a single ``meta-taxon''.

Let $T^*$ (respectively, $T^{**}$) be the set of trees obtained by taking the trees on $X'$ and contracting all the $X_2$ (respectively, $X_1$) taxa into a single meta-taxon $W_2$ (respectively, $W_1$).
Note that contracting in this way cannot increase the treewidth of $D$ and that $1 \leq |X_i| < |X|$ for $i\in\{1,2\}$. Hence, by induction supertrees of $T^*$ and $T^{**}$ exist. Finally, construct supertree $S'$ with split $X_1|X_2$ from two supertrees for $T^*$ and $T^{**}$ by adding an edge between $W_1$ and $W_2$ and afterwards suppressing $W_1$ and $W_2$. (The function of $W_1$ and $W_2$ was precisely to ensure that we would know how to glue the two separately constructed supertrees together).

To see the second claim note that since $S'$ is a supertree of $T_1,...,T_k$ restricted $X\setminus \{u\}$ we only have to show that quartets of $T_1,...,T_k$ that contain taxon $u$ are displayed by $S$. So w.l.o.g. let $a \in X_1, b,c \in X_2$. Then if quartet $au|bc$ is displayed by some input tree $T$ it is also clearly displayed by the supertree $S$. We claim quartets $ub|ac$ or $uc|ab$ cannot exist in any of the input trees. These two quartets are the same up to relabeling so let's consider quartet $ub|ac$ induced by some tree $T$ sitting inside $D$. Then $p(u,b)$ and $p(a,c)$ are edge-disjoint and contain no taxa. As argued before $p(a,c)$ must pass through $v$. But since $(u,v)$ is an edge it follows that it must belong to the same tree $T$, and therefore $v$
also lies on the path $p(u,b)$. But then it is not possible that $T$ displays $ub|ac$, contradiction.



\noindent
\textbf{Case(iii)} is when both $u$ and $v$ are inner nodes. We could have that $p(u,v)$ is an edge, in which case $u$ and $v$ are inner nodes of the same tree, or we could have that $p(u,v)$ contains a single taxon $t$. Note that in the latter case $u$ and $v$ are inner nodes of two different trees and taxon $t$ must have degree 2 in $D$ due to the minimal adjacency of $F_1$ and $F_2$. 
The argument for $\{u,v\}$ being a separator of $D$ goes through in this case as well regardless of $p(u,v)$ being an edge or a path containing a single taxon $t$. We again denote by $X_1$ and $X_2$ the two sets of taxa that emerge from splitting $D$ by removing $u$ and $v$ (and $t$ if it exists on $(u,v)$). 

\begin{figure}[h]
 \begin{center}
  \includegraphics[width=10cm]{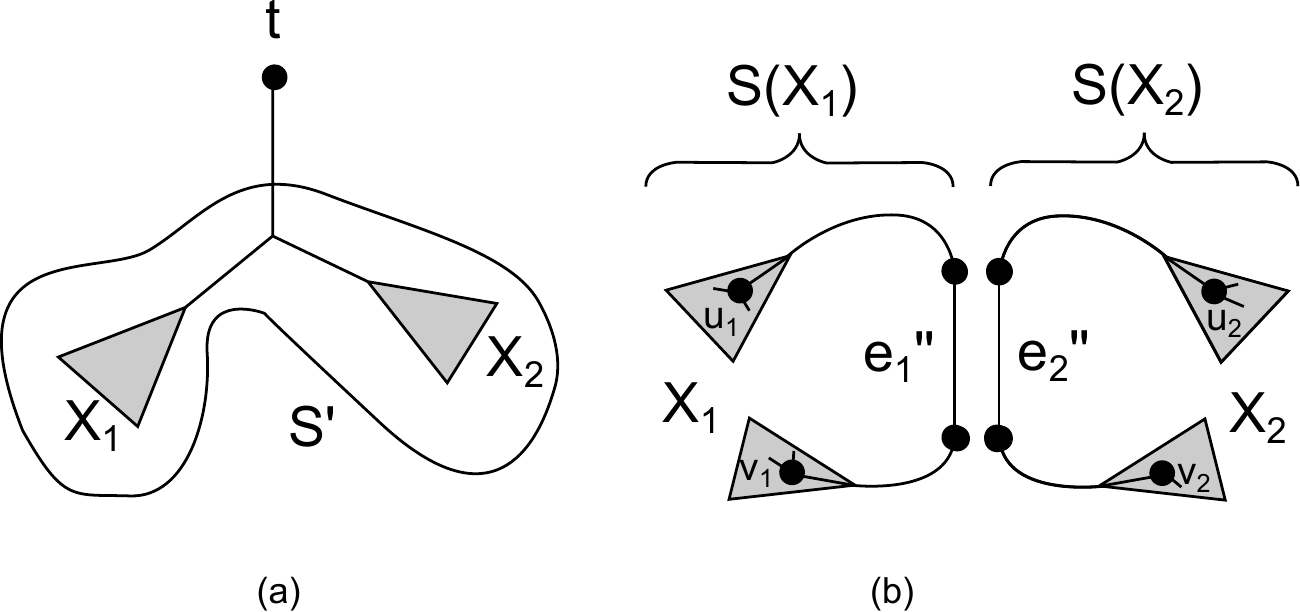}
 \end{center} 
 \caption{(a) A supertree constructed in case (iii) when there exists a taxon $t$ on the common boundary of the two faces. (b) Construction of a supertree in case (iii) when the common boundary of the two faces is a single edge.}
 \label{fig:K4separatorCase3}
\end{figure}

\textbf{Subcase 1}. Consider first the subcase when some taxon $t \in p(u,v)$. As before we have to show that  there exists some $S'$, a supertree of $T_1,...,T_k$ restricted to $X':= X \setminus \{t\}$ with split $X_1|X_2$, and that the supertree $S$ as shown in Figure \ref{fig:K4separatorCase3}(a) displays all quartets induced by $T_1,...,T_k$. 

The proof of the first part (i.e. that a suitable $S'$ exists) is almost exactly the same as in case (ii). There were three places in which we used the fact that $u$ was a leaf. We now show that those statements hold also when $u$ is an inner node and $t$ is a taxon on path $p(u,v)$. We saw that in case (ii) both $X_1$ and $X_2$ were nonempty and of strictly smaller cardinality than $X$. That they are strictly smaller than $X$ is the case here as well since $t\in X$ but $t \notin X_1$ and $t \notin X_2$. That they are nonempty in this case also holds. Consider face $F_1$. Let $u$ be a node of some tree $T_1$ and $v$ a node of some other tree $T_2$. Then the path from $u$ to $v$ that follows the part of the boundary $F_1$ not shared by $F_2$, is a path between two vertices of different trees and so must contain some taxon $a \neq t$ (which is also in $X_1$). The same argument holds for $F_2$ and $X_2$.  Another thing we have to show here is that all input trees $T_1,...,T_k$ restricted to $X'$ respect split $X_1|X_2$
(i.e. each tree contains an edge $e$ inducing a split $A|B$ where $A \subseteq X_1$ and $B \subseteq X_2$). Suppose there exists a quartet $ab|cd$ with $a,c \in X_1$ and $b,d \in X_2$ in some input tree $T$. Since removing $u$ and $v$ from $D$ disconnects it, it must be that two paths $p(a,b)$ and $p(c,d)$ had to pass through either $u$ or $v$. W.l.o.g. let $u \in p(a,c)$ and $v \in p(b,d)$. Furthermore paths $p(a,b)$ and $p(c,d)$ are edge-disjoint in $D$ and belong to the same tree $T$. But this is impossible since $u$ and $v$ belong to different trees in this case. Contradiction. So we conclude in this case too that every input tree restricted to $X'$ respects split $X_1|X_2$, and hence the contraction of $X_1$ and $X_2$ into meta-taxa works exactly
as described in case (ii). Hence, $S'$ indeed exists, can be constructed and has split $X_1|X_2$.

Next, we claim that $S$ as shown in Figure \ref{fig:K4separatorCase3}(a) displays all quartets induced by $T_1,...,T_k$. As before, since $S'$ is a supertree of $T_1,...,T_k|X'$ we only need to check the quartets induced by $T_1,...,T_k$ that contain taxon $t$. W.l.o.g. let $a \in X_1$ and $b,c \in X_2$. There are three possible topologies $at|bc, bt|ac, ct|ab$. As before, $at|bc$ is an easy case since if it appears in some $T$ in $D$ it clearly also appears in $S$, while topologies $bt|ac$ and $ct|ab$ are the same up to relabeling. So consider $bt|ac$ and suppose it is displayed by some tree $T$. Paths $p(b,t)$ and $p(a,c)$ are edge-disjoint and since the degree of $t$ is 2 in $D$ we have that $p(b,t)$ has to contain either $u$ or $v$. W.l.o.g. let $u \in p(b,t)$. Now, since $\{u,v\}$ is a separator of $D$ and $a \in X_1$ while $c \in X_2$, we have that either $u \in p(a,c)$ or $v \in p(a,c)$. If $u \in p(a,c)$ we have that paths $p(b,t)$ and $p(a,c)$ both contain node $u$, a contradiction on $bt|ac$ being displayed by $T$. If $v \in p(a,c)$ then edge $(v,t)$ and $(u,t)$ must both belong to the same tree $T$, which contradicts our earlier observation that $u$ and $v$ are necessarily in different trees.

\textbf{Subcase 2}. The last thing to consider is the subcase when $(u,v)$ is an edge while both $u$ and $v$ are inner nodes (necessarily of the same tree $T$). Let $X_1$ and $X_2$ be two disjoint sets of taxa that result from splitting $D$ after removing $u$ and $v$. We claim that $|X_1| \geq 2$ and $|X_2| \geq 2$. This follows directly from $u,v \in T$: any cycle that links them together must leave the tree $T$ via some taxon $a$ and re-enter it via a (necessarily different) taxon $b$. Since $u$ and $v$ belong to both faces $F_1$ and $F_2$ it follows that the boundaries of these two faces must each contain (at least) two taxa. The two taxa on the boundary of (w.l.o.g) $F_1$ are still in the same connected component after deletion of $\{u,v\}$, but are not in
the same connected component as the taxa from the boundary of $F_2$, so $|X_1| \geq 2$ and $|X_2| \geq 2$.

Now we claim that the tree shown in Figure \ref{fig:K4separatorCase3}(b) is a supertree of $T_1,...,T_k$. Let's first explain what that image means. Note that apart from the tree $T$ in which the internal edge $e=(u,v)$ can be found, all other trees have taxa sets either completely contained inside $X_1$ or completely contained inside $X_2$. This is the case because otherwise there would be a path from some element in $X_1$ to some element in $X_2$, contradicting the fact that $\{u,v\}$ is a separator. The idea is to cut $T$ into two parts, one on $X_1$, one on $X_2$, recursively build supertrees of $T_1,...,T_k|X_1$ and $T_1,...,T_k|X_2$ and join them as indicated in the figure.

Now, consider the display graph $D$. Suppose we delete the edge $e=(u,v) \in T$, and replace it with two edges $e_1=(u_1,v_1)$ and $e_2=(u_2,v_2)$ (where $u_i$ and $v_i$ are $u$ and $v$ duplicated). Because $\{u,v\}$ is a separator, this creates two disjoint display graphs, one on $X_1$ and one on $X_2$. These are minors of the original display graph so have treewidth at most 2, and they are smaller instances of the problem. So by induction supertrees of these smaller instances exist. Let $S(X_1)$ be a supertree on $X_1$ and $S(X_2)$ be a supertree on $X_2$. All trees except $T$ will be displayed by the disjoint union of $S(X_1)$ and $S(X_2)$, because only $T$ has taxa from both $X_1$ and $X_2$. What is left to explain is how to glue $S(X_1)$ and $S(X_2)$ into a supertree $S$ such that $S$ displays $T$ as well. 

Note that $S(X_i)$ contains an image of edge $e_i$. The image need not be an edge in $S(X_i)$, it could also be a path, whose endpoint we denote by $u_i$ and $v_i$ in Figure \ref{fig:K4separatorCase3}(b). Take any edge on path $p(u_i, v_i)$, call it $e'_i$, and subdivide it twice to create two adjacent degree-2 vertices; let $e''_i$ be the edge between them. Now, by identifying $e''_1$ and $e''_2$ we ensure that we get a supertree that displays (all the quartets in) $T$, as well as all the other trees.

This completes the case analysis. Polynomial time is achieved because all relevant operations (recognizing whether a graph has treewidth at most 2, finding a planar embedding, finding two minimally adjacent faces, finding the separator $\{u,v\}$, and
all the various tree manipulation operations) can easily be performed in (low-order) polynomial time.


\end{proof}

\section{Beyond treewidth 2}

Two incompatible quartets induce a display graph with treewidth 3, so treewidth 3 cannot guarantee compatibility. However, it is natural to ask whether treewidth 3 guarantees compatibility if the number of input trees becomes sufficiently large.
Unfortunately, the answer to that question is no. Namely, for any number of trees there exists a compatible instance with $tw(D)=3$ and an incompatible instance with $tw(D)=3$, as we now demonstrate. Figure \ref{fig:gray_comp_tw3} shows the display graph of $k$ trees with leaves denoted as black dots and vertices of $K_4$ minors with red dots (note that some leaves, for example $z$, can also be a vertex of a $K_4$ minor). Note that vertices $a,b,c,z$ form a $K_4$ minor in $D(T_1,T_2,T_3)$, vertices $b,c,d,q$ form a $K_4$ minor in $D(T_2,T_3,T_4)$, vertices $d,e,f,s$ form a $K_4$ minor in $D(T_4,T_5,T_6)$ and so on. Now note that all those $K_4$ minors are attached together by a sequence of series and parallel compositions inside $D(T_1,...,T_k)$. So we can conclude that the treewidth of the display graph of $k$ trees as shown in figure \ref{fig:gray_comp_tw3} is 3. (Equivalently, we can describe a tree decomposition in which all bags have size at most 4). Compatibility of this instance can be
verified without too much difficulty (details omitted).

Now we need to show the same for an incompatible instance. In Figure \ref{fig:gray_incomp_tw3} trees $T_1,T_2,T_3$ are incompatible, thus the whole instance is incompatible. Furthermore, trees $T_4,...,T_k$ are chosen to be the same as in Figure \ref{fig:gray_comp_tw3}, so are compatible and $tw(D(T_4,...,T_k))=3$. We have verified that $tw(D(T_1,T_2,T_3))=3$. Since $D(T_1,T_2,T_3)$ and $D(T_4,...,T_k)$ are attached in series to form $D(T_1,...,T_k)$ we conclude $tw(D(T_1,...,T_k))=3$. 

It is not difficult to generalize these constructions for any treewidth higher than 3, and any number of trees.


\begin{figure}[h]
	\centering
		\includegraphics[width=0.8\textwidth]{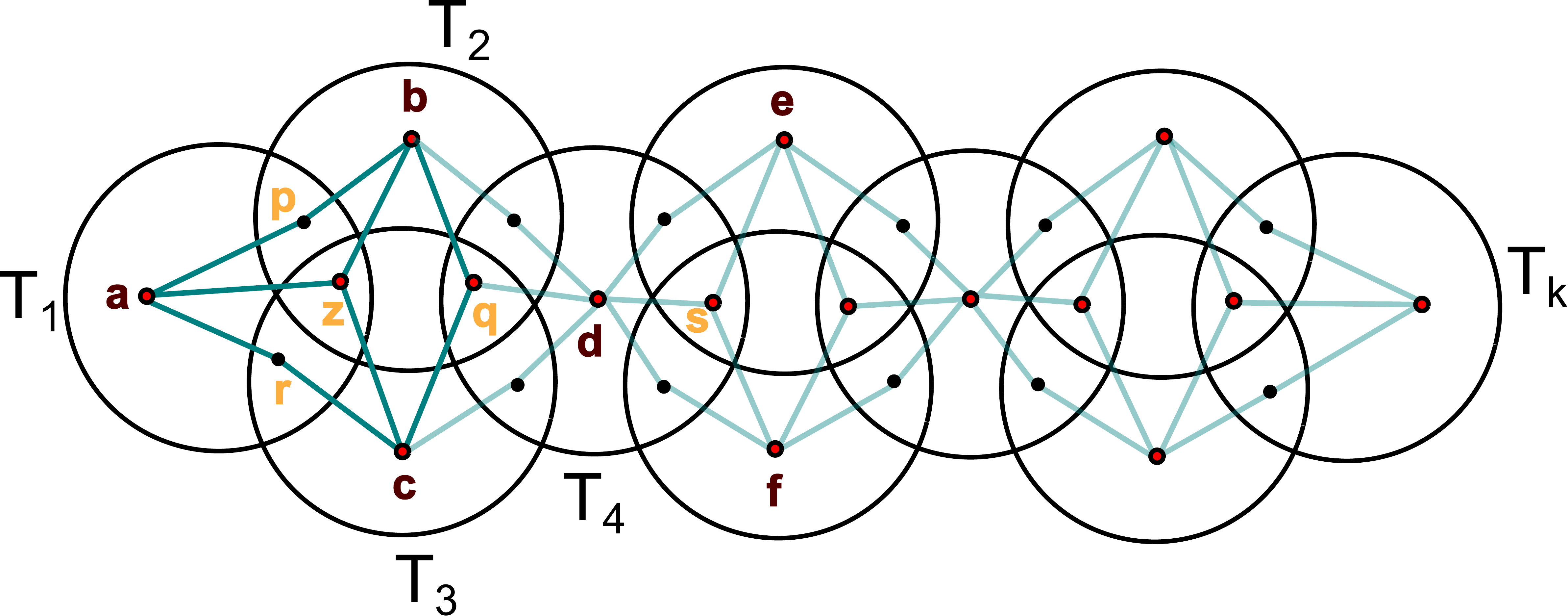}
	\caption{Display graph of an instance with $k$ input trees. Red (larger) vertices are inner nodes while black (smaller) vertices are leaves. The treewidth of $D$ is 3 and the instance is compatible.}
	\label{fig:gray_comp_tw3}
\end{figure}

\begin{figure}[h]
	\centering
		\includegraphics[width=0.8\textwidth]{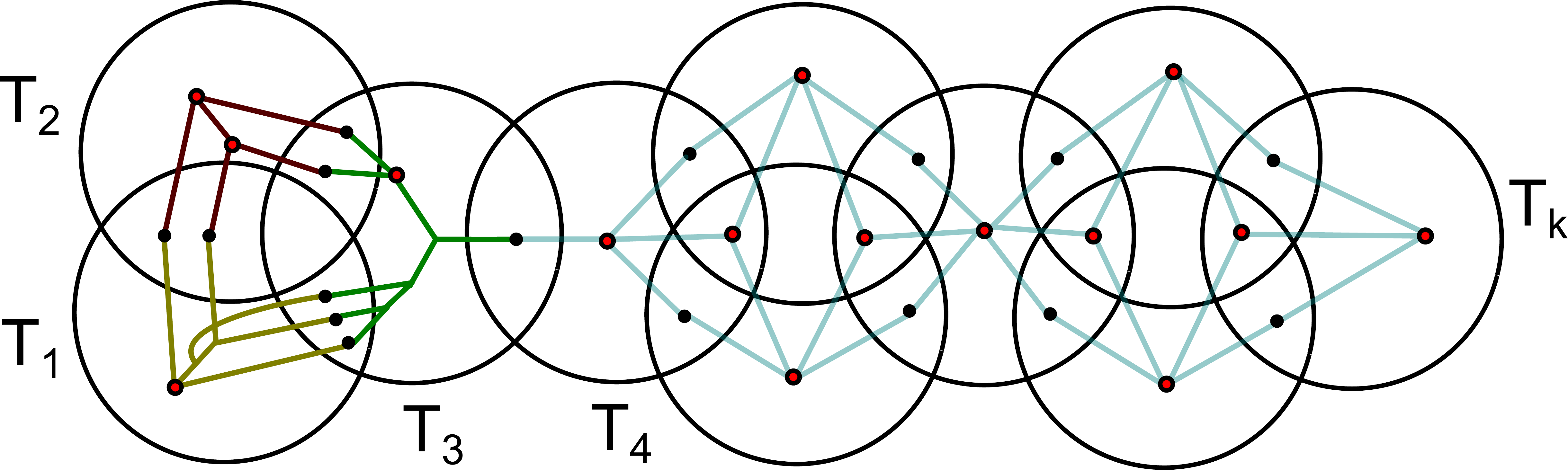}
	\caption{Display graph of an instance with $k$ input trees. The treewidth of $D$ is 3 and the instance is incompatible.}
	\label{fig:gray_incomp_tw3}
\end{figure}



\section{Conclusion}

\begin{figure}[h]
	\centering
		\includegraphics[width=0.40\textwidth]{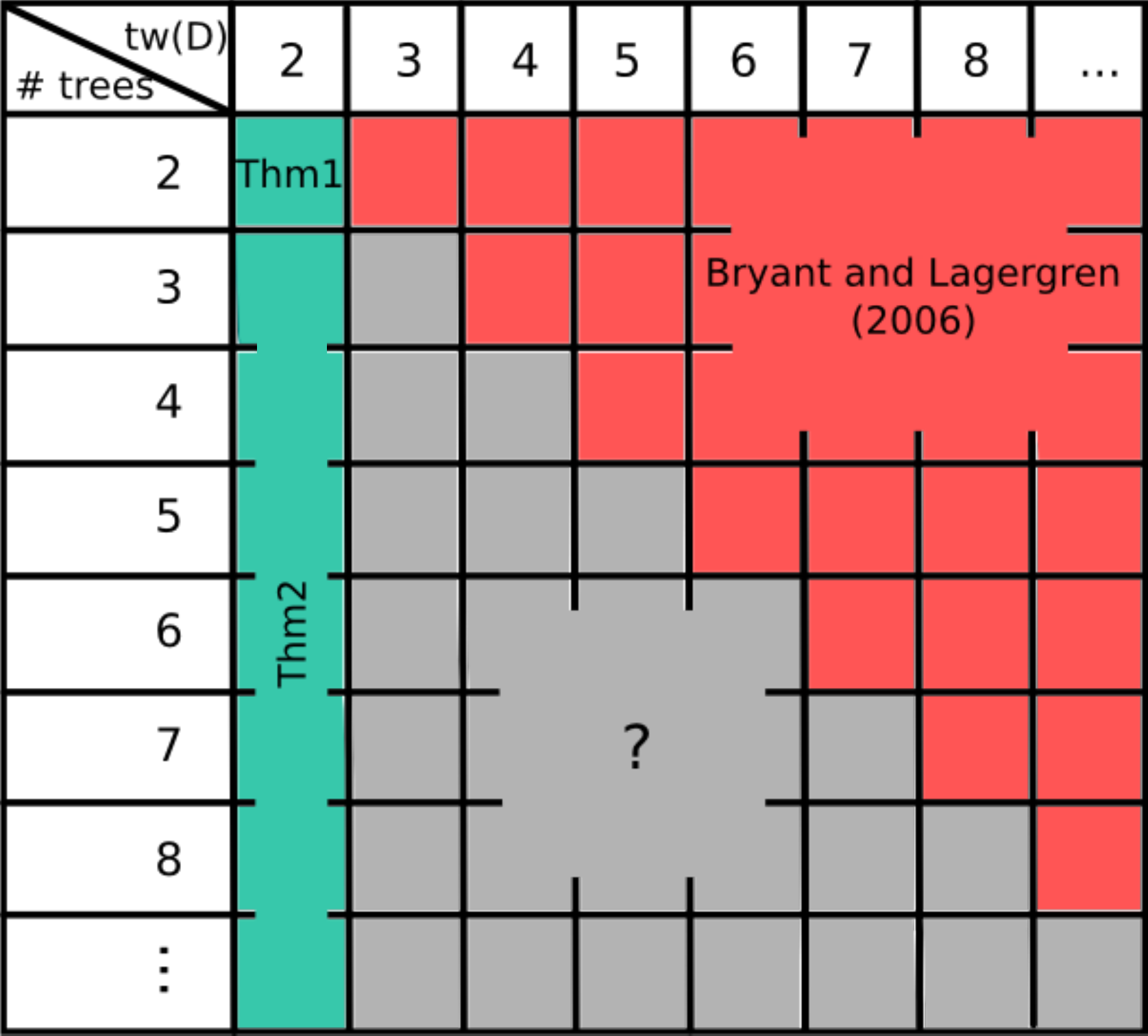}
	\caption{The green (respectively, red) area shows which combinations of  (number of input trees, treewidth of display graph) are always compatible (respectively, incompatible). The grey area indicates that both compatible and incompatible instances exist for this combination of parameters.}
	\label{fig:table}
\end{figure}

Figure \ref{fig:table} summarizes our results. The red area is due to result of Bryant and Lagergren which proves that that any instance on $k$ trees whose display graph has treewidth strictly greater than $k$ must be incompatible. The green area is due to our result. What we are left with is the grey area in which (as demonstrated by the constructions in the previous section) we cannot conclude anything about compatibility of the instances based only on treewidth of the display graph and the number of trees, at least not with the current results. 
An obvious open questions is whether existing
characterizations (such as legal triangulations \cite{vakati2011graph}) can be specialized to yield simple and efficient combinatorial algorithms in the case of treewidth 3 or higher.


\bibliographystyle{plain}
\bibliography{compatibility}

\appendix

\section{Existence of minimally adjacent faces}


\begin{observation}
\label{obs:adj}
Let $\cT$ be a cleaned up, non-empty set of unrooted binary trees on $X$ such that $tw(D(\cT)) = 2$. Consider any planar embedding of $D(\cT)$. Then there exist two distinct faces $F_1, F_2$ in $D(\cT)$ such that $F_1$ and $F_2$ are adajcent and neither is equal to the outer face.
\end{observation}
\begin{proof}
Without loss of generality we prove this for the case when $D(\cT)$ is connected. Recall that all vertices in $D(\cT)$ have degree at least 2, and at least one vertex has at least degree 3 (due to the existence of internal nodes). Hence, by the handshaking lemma, the number of edges in $D(\cT)$ is strictly larger than the number of vertices, and thus we can use Euler's formula to conclude that $D(\cT)$ has at least 3 faces. One of these is the outer face, so $D(\cT)$ has at least two faces not equal to the outer face. Hence, $D(\cT)$ contains at least two simple cycles. If any two simple cycles have a common edge, then we are done, so let us assume that all simple cycles in $D(\cT)$ are edge disjoint (and chordless). However, this is not possible due to the fact that in every simple cycle at least two vertices have degree 3 or higher 
and the fact that all vertices in the graph have degree at least 2. (In particular, a simple cycle can never act as a ``sink'' to absorb excess degree, and there are also no leaves to fulfil this function.)
\end{proof}

\begin{observation}
\label{obs:adj2}
Let $\cT$ be a cleaned up, non-empty set of unrooted binary trees on $X$ such that $tw(D(\cT)) = 2$. Consider any planar embedding of $D(\cT)$. Let $e$ be a cut-edge of $D(\cT)$, and let $D_1$, $D_2$ be the two components obtained by deleting $e$. Then both $D_1$ and $D_2$ have their own pair of adjacent faces, neither equal to the outer face.
\end{observation}
\begin{proof}
This is a simple adaptation of the previous proof. Deleting $e$ reduces the degree of two vertices by exactly one, and all other degrees are unchanged. So $D_1$ and $D_2$ both contain at most one vertex of degree 1. 
From the previous ``sink'' observation we see that both $D_1$ and $D_2$ must contain two simple cycles with intersecting edges, and we are done.
\end{proof}


Recall the definition of minimal adajcency from the main text.

\begin{lemma}
\label{lemma:inclmin}
 Let $\cT$ be a cleaned up, non-empty set of unrooted binary trees on $X$ such that $tw(D(\cT)) = 2$. Consider any planar embedding of $D(\cT)$. Then there exist two distinct faces $F_1, F_2$ in $D(\cT)$ such that $F_1$ and $F_2$ are \underline{minimally} adjacent and neither is equal to the outer face. Also, these can be found in polynomial time.
\end{lemma}
\begin{proof}
Fix any planar embedding of $D(\cT)$. Let $G$ be the graph whose vertices are the faces of $D(\cT)$ (including the outer face) and whose edges are the adjacency relation on those faces. We label each face $F$ of $G$ with the length of a shortest path in $G$ from $F$ to the outer face. Clearly, the outer face has label 0. Let $k$ be the maximum label ranging over all faces. We select a pair of distinct faces $(F_1, F_2)$ such that (1) $F_1$ has label $k$; (2) $F_2$ is adjacent to $F_1$; (3) $F_2$ has the largest label ranging over all faces that are adjacent to a face with label $k$. By Observation \ref{obs:adj}, $F_1$ and $F_2$ both have label at least 1. Note also that the label of $F_2$ is either $k-1$ or $k$. Clearly, $F_1$ and $F_2$ both satisfy property (1) of minimal adjacency. 

Consider now the sequence of vertices and edges $v_1, e_1, v_2, e_2, \ldots, v_n = v_1$ that define the boundary of face $F_1$. Observe that with the exception of $v_1 = v_n$ all vertices on the boundary are distinct. This is because, if the boundary of the face intersects with itself, it creates a new face $F_3$ ``inside'' $F_1$ whose shortest path to the outer face is strictly larger than $k$, contradicting the minimaliity of $k$. Hence, $B(F_1)$ is a simple cycle. From this it follows that $B(F_1) \cap B(F_2)$ is a subgraph of a simple cycle. In particular, it can be (a) a simple cycle or (b) a set of one or more paths (where some of the paths might have length 0). We show that (a) cannot happen.  To see this, observe that (a) can only happen if $B(F_1) \subseteq B(F_2)$. From the degree constraints mentioned earlier the simple cycle defining $F_1$ contains at least 2 vertices of degree 3 or higher, in $D(\cT)$. These two vertices $u_1, u_2$ generate paths that cannot enter the interior of $F_1$, because they would then necessarily slice $F_1$ up into smaller faces. Moreover, there cannot exist a path from $u_1$ to $u_2$ that avoids $B(F_1)$, because this would imply the existence of a third face $F_3$ adjacent to $F_1$, such that $B(F_1) \cap B(F_3)$ contains an edge not in $B(F_2)$. In particular, this would contradict $B(F_1) \subseteq B(F_2)$.  For a similar reason, the generated paths cannot re-intersect with $B(F_1)$. Careful analysis shows that the only remaining possibility is that $F_2$ is the outer face, contradicting the fact that the label of $F_2$ is at least 1. Hence we conclude that (a) is not possible, and that (b) must hold.

We now establish property (2) of minimal adjacency. In particular we show that $B(F_1) \cap B(F_2)$ has a single component. By the definition of adjacency, and the fact that (b) holds, at least one component in $B(F_1) \cap B(F_2)$ is a path $P$ on one or more edges. Clearly, the two endpoints of $P$ must (in $D(\cT)$) have degree 3 or higher, otherwise $P$ could be extended further. Without loss of generality consider the lower endpoint $u$. Let $e$ be an edge incident to $u$ (in $D(\cT)$) that is not in $B(F_1) \cap B(F_2)$ but which is incident to $F_1$ (such an edge must exist). The second face incident to $e$ cannot be $F_2$, because otherwise $e$ would be in $P$, so it must be some other face $F_3 \neq F_2$. Suppose there exists a path $P' \neq P$ in $B(F_1) \cap B(F_2)$. (Possibly, $P'$ is a single vertex). In this case it is possible to draw a closed curve that passes through $P$ and $P'$ and such that the only face interiors that it intersects with, are those of $F_1$ and $F_2$ (see Figure \ref{fig:curve}). Informally this means that face $F_3$ is entirely ``enclosed'' by $F_1$ and $F_2$.

\begin{figure}[h]
\begin{center}
\includegraphics[width=5cm]{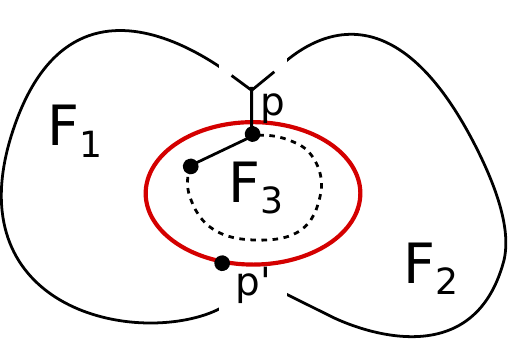}
\end{center}
\caption{\emph{If $B(F_1) \cap B(F_2)$ consists of two or more components $P, P'$ then it is possible to draw a curve (shown in red) completely enclosing a third face $F_3$, yielding a contradiction on the choice of $F_1$ and/or $F_2$.}}
\label{fig:curve}
\end{figure}

More precisely, it means that any (shortest) path in $G$ from $F_3$ to the outer face must pass through $F_1$ or $F_2$. If such a shortest path travels via $F_1$, then the label
of $F_3$ is at least $k+1$, contradicting the maximality of $k$. If it travels via $F_2$, then it has label $k$ or $k+1$. The latter is clearly a contradiction, but also the former because this contradicts our earlier choice of $F_2$ (i.e. we should have chosen $F_3$ instead of $F_2$). Hence, $B(F_1) \cap B(F_2)$ indeed consists of a single path $P$ (containing at least one edge).

It remains to prove property (3). We have already established that the endpoints of $P$ have degree 3 or more in $D(\cT)$. If $P$ has no interior vertices, or all interior vertices of $P$ have degree 2 in $D(\cT)$, we are done. So suppose $P$ contains an interior vertex $u$ of degree 3 or more in $D(\cT)$. Let $e$ be an edge incident to $u$ (in $D(\cT)$) that is not in $B(F_1) \cap B(F_2)$. Clearly, $e$ starts a path that extends into the interior of $F_1$ or $F_2$. If $e$ is not a cut-edge then the path it starts must re-intersect with the boundary of $F_1$ or $F_2$, but this causes a face to be partitioned into smaller pieces, which is not possible. Hence, $e$ must be a cut-edge. From Observation \ref{obs:adj2} deleting $e$ yields two or more adjacent faces that are entirely ``enclosed'' by $F_1$ or $F_2$. If they are enclosed by $F_1$ then they both have label $k+1$, which is a contradiction. If they are enclosed by $F_2$, and $F_2$ has label $k$, the same contradiction is obtained. If they are enclosed by $F_2$, and $F_2$ has label $k-1$, then they both have label $k$, contradicting the fact that we chose $F_2$ in the first place.

Polynomial time is assured since recognition of treewidth 2, planar embeddings and determination of the labels can all be computed
in (low-order) polynomial time.

\end{proof}


\end{document}